\documentclass[a4paper,11pt]{article}
\pdfoutput=1
\usepackage{fullpage}
\usepackage{array}
\usepackage{xspace}
\usepackage[utf8]{inputenc}
\usepackage[T1]{fontenc}
\usepackage{graphicx}
\usepackage{amsmath}
\usepackage{float}
\usepackage{amsfonts}
\usepackage{algorithm}      
\usepackage{algpseudocode}  
\usepackage[dvipsnames,usenames]{color}
\usepackage[colorlinks=true,urlcolor=Blue,citecolor=Green,linkcolor=BrickRed]{hyperref}
\usepackage[usenames,dvipsnames]{xcolor}
\usepackage{authblk}
\usepackage{amsthm}
\usepackage{bigstrut}
\usepackage[noadjust]{cite}
\usepackage{todonotes}
\usepackage{thmtools,thm-restate}

\newcommand{\Oh}{\mathcal{O}}
\newcommand{\nca}{\mathsf{NCA}}
\newcommand{\degree}{\mathsf{deg}}
\newcommand{\inner}{\mathsf{inner}}
\newcommand{\weight}{\mathsf{weight}}
\newcommand{\nodelabel}{\ell}
\def\polylog{\operatorname{polylog}}

\makeatletter
\newcommand*{\bdiv}{%
  \nonscript\mskip-\medmuskip\mkern5mu%
  \mathbin{\operator@font div}\penalty900\mkern5mu%
  \nonscript\mskip-\medmuskip
}
\makeatother

\newtheorem{theorem}{Theorem}
\newtheorem{lemma}[theorem]{Lemma}
\newtheorem{proposition}[theorem]{Proposition}

\title{Better Labeling Schemes for Nearest Common Ancestors\\ through Minor-Universal Trees}

\author[1,2]{Pawe{\l} Gawrychowski}
\author[2]{Jakub Łopuszański}
\affil[1]{University of Haifa, Israel}
\affil[2]{University of Wrocław, Poland}

\date{}

\begin{document}

\maketitle

\begin{abstract}
Preprocessing a tree for finding the nearest common ancestor of two nodes
is a basic tool with multiple applications. Quite a few linear-space constant-time
solutions are known and the problem seems to be well-understood.
This is however not so clear if we want to design a labeling scheme.
In this model, the structure should be
distributed: every node receives a distinct binary string, called its label, so that
given the labels of two nodes (and no further information about the topology
of the tree) we can compute the label of their nearest common ancestor.
The goal is to make the labels as short as possible. Alstrup, Gavoille,
Kaplan, and Rauhe [Theor. Comput. Syst. 37(3):441-456 2004] showed
that $\Oh(\log n)$-bit labels are enough, with a somewhat large constant.
More recently, Alstrup, Halvorsen, and Larsen [SODA 2014] refined this to
only $2.772\log n$, and provided a lower bound of $1.008\log n$.

We connect the question of designing a labeling
scheme for nearest common ancestors to the existence of a tree, called a minor-universal
tree, that contains every tree on $n$ nodes as a topological minor.
Even though it is not clear if a labeling scheme
must be based on such a notion, we argue that all already existing schemes
can be reformulated as such. Further, we show that this notion allows us to
easily obtain clean and good bounds on the length of the labels.
As the main upper bound, we show that $2.318\log n$-bit labels are enough.
Surprisingly, the notion of a minor-universal tree for binary trees
on $n$ nodes has been already used in a different context by Hrubes et al. [CCC 2010],
and Young, Chu, and Wong [J. ACM 46(3):416-435, 1999] introduced a very closely
related (but not equivalent) notion of a universal tree.
On the lower bound side, we show that any minor-universal tree for trees
on $n$ nodes must contain at least $\Omega(n^{2.174})$ nodes. This highlights a natural
limitation for all approaches based on defining a minor-universal tree.
Our lower bound technique also implies that a universal tree in the sense of Young
et al. must contain at least $\Omega(n^{2.185})$ nodes, thus dramatically improves their
lower bound of $\Omega(n\log n)$.
We complement the existential results with a generic
transformation that allows us, for any labeling scheme for nearest common
ancestors based on a minor-universal tree, to decrease the query time to
constant, while increasing the length of the labels only by lower order terms.
\end{abstract}

\thispagestyle{empty}
\clearpage
\setcounter{page}{1}

\section{Introduction}

A labeling scheme assigns a short binary string, called a label, to each node in a network, so that a function on two nodes
(such as distances, adjacency, connectivity, or nearest common ancestors) can be computed by
examining their labels alone. We consider designing such scheme for finding the nearest common
ancestor (NCA) of two nodes in a tree. More formally,
given the labels of two nodes of a rooted tree, we want to compute the label of their nearest
common ancestor (for this definition to make sense, we need to explicitly require that the labels
of all nodes in the same tree are distinct).

Computing nearest common ancestors is one of the basic algorithmic questions that one can
consider for trees. Harel and Tarjan~\cite{HarelT84} were the first to show how to preprocess
a tree using a linear number of words, so that the nearest common ancestor of any two nodes can be found
in constant time. Since then, quite a few simpler solutions have been found, such as the one
described by Bender and Farach-Colton~\cite{BenderF00}. See the survey by Alstrup et al.~\cite{AlstrupGKR02}
for a more detailed description of these solutions and related problems.

While constant query time and linear space might seem optimal, some important applications 
such as network routing require the structure to be distributed. That is, we might want to 
associate some information with every node of the tree, so that the nearest common ancestor
of two nodes can be computed using only their stored information. The goal is to distribute
the information as evenly as possible, which can be formalized by assigning a binary string, called a label, to
every node and minimizing its maximum length. This is then called a labeling scheme for nearest
common ancestors. Labeling schemes for multiple other queries in trees have been considered, 
such as distance~\cite{Peleg00,gavoille2004distance,alstrup2005labeling,alstrup2015distance,gavoille2007distributed},
adjacency~\cite{alstrup2015optimal,alstrup2002small,bonichon2007short}, ancestry~\cite{abiteboul2006compact,fraigniaud2010compact},
or routing~\cite{thorup2001compact}.
While we focus on trees, such questions make sense and have been considered also for more general classes of 
graphs~\cite{alstrup2005labeling,abiteboul2006compact,fischer2009short,fraigniaud2010compact,AHL14,alstrup2015adjacency,alstrup2015optimal,petersen2015near,alstrup2015distance,alstrup2016simpler,GawrychowskiKU16,GawrychowskiU16,AbboudGMW17,AlonN17,KatzKKP04}.
See~\cite{rotbart2016new} for a survey of these results.

Looking at the structure of Bender and Farach-Colton, converting it into a labeling scheme with short label is not trivial, as we
need to avoid using large precomputed tables that seem essential in their solution.
However, Peleg~\cite{Peleg05} showed how to assign a label consisting of $\Oh(\log^{2}n$) bits
to every node of a tree on $n$ nodes, so that given the labels of two nodes we can return
the predetermined name of their nearest common ancestor. He also showed that this
is asymptotically optimal. Interestingly, the situation changes quite dramatically if
we are allowed to design the names ourselves. That is, we want to assign a distinct name
to every node of a tree, so that given the names of two nodes we can find the name
of their nearest common ancestor (without any additional knowledge about the structure
of the tree). This is closely connected to the implicit representations of graphs considered
by Kannan et al.~\cite{Kannan}, except that in their case the query was adjacency.
Alstrup et al.~\cite{AlstrupGKR02} showed that, somewhat surprisingly, this is enough
to circumvent the lower bound of Peleg by designing a scheme using labels consisting
of $\Oh(\log n)$ bits. They did not calculate the exact constant, but later experimental
comparison by Fischer~\cite{fischer2009short} showed that, even after some tweaking, in the worst case it is around 8.
In a later paper, Alstrup et al.~\cite{AHL14} showed an NCA labeling scheme with labels of length 
$2.772\log n + \Oh(1)$\footnote{In this paper, $\log$ denotes the logarithm in base 2.} and proved that any such scheme needs labels of length at least
$1.008\log n - \Oh(1)$. The latter non-trivially improves an immediate lower bound of $\log n+\Omega(\log\log n)$
obtained from ancestry. They also presented an improved scheme for binary
trees with labels of length $2.585 \log n + \Oh(1)$.

The scheme of Alstrup et al.~\cite{AHL14} (and also all previous schemes) is based on the notion of
heavy path decomposition. For every heavy path, we assign a binary code to the root of each 
subtree hanging off it. The length of a code should correspond to the size of the subtree, so that
larger subtrees receive shorter codes. Then, the label of a node
is the concatenation of the codes assigned to the subtrees rooted at its light ancestors,
where an ancestor is light if it starts a new heavy path.
These codes need to be appropriately delimited, which makes the whole construction
(and the analysis) somewhat tedious if one is interested in optimizing the final
bound on the length.

\subsection{Our Results}

Our main conceptual contribution is connecting labeling schemes for nearest common ancestors to the
notion of minor-universal trees, that we believe to be an elegant approach for obtaining simple
and rather good bounds on the length of the labels, and in fact allows us to obtain significant
improvements.
It is well known that some labeling problems have a natural and clean connection
to universal trees, in particular these two views are known to be equivalent for adjacency~\cite{Kannan}
(another example is distance~\cite{FGNW16}, where the notion of a universal tree gives
a quite good but not the best possible bound). It appears that no such connection has
been explicitly mentioned in the literature for nearest common ancestors so far.
Intuitively, a minor-universal tree for trees on $n$ nodes, denoted $U_{n}$, is a rooted tree, such that the nodes of
any rooted tree $T$ on $n$ nodes can be mapped to the nodes of $U_{n}$ as to preserve the
NCA relationship. More formally, $T$ should be a topological minor of $U_{n}$,
meaning that $U_{n}$ should contain a subdivision of $T$ as a subgraph, or in other words
there should exists a mapping $f : T \rightarrow U_{n}$ such that $f(\nca(u,v))=\nca(f(u),f(v))$
for any $u,v\in T$.
This immediately implies
a labeling scheme for nearest common ancestors with labels of length $\log |U_{n}|$,
as we can choose the label of a node $u\in T$ to be the identifier of the node of $U_{n}$
it gets mapped to (in a fixed mapping), so small $U_{n}$ implies short labels.
In this case, it is not clear if a reverse connection holds. Nevertheless, all previously
considered labeling schemes for nearest common ancestors that we are aware of
can be recast in this framework.

The notion of a minor-universal tree has been independently considered before in different contexts.
Hrubes et al.~\cite{HrubesWY10} use it to solve a certain
problem in computational complexity, and construct a minor-universal tree of size
$n^{4}$ for all ordered binary trees on $n$ nodes (we briefly discuss how our
results relate to ordered trees in Appendix~\ref{sec:ordered}). Young et al.~\cite{YoungCW99} 
introduce a related (but not equivalent) notion of a universal tree, where instead
of a topological minor we are interested in minors that preserve the depth modulo
2, to study a certain question on boolean functions,
and construct such universal tree of size $\Oh(n^{2.376})$ for all trees on $n$ nodes.

Our technical contributions are summarized in Table~\ref{tbl:contribution}. 
The upper bounds are presented in Section~\ref{sec:upper} and should be compared
with the labeling schemes of Alstrup et al.~\cite{AHL14}, that imply a minor-universal tree of size 
$\Oh(n^{2.585})$ for binary trees, and $\Oh(n^{2.772})$ for general (without restricting the degrees)
trees, and $\Oh(n^{2.585})$, and the explicit construction
of a minor-universal tree of size $\Oh(n^{4})$ for binary trees given by Hrubes et al.~\cite{HrubesWY10}.
The lower bounds are described in Section~\ref{sec:lower}.
We are aware of no previously existing lower bounds on the size of a minor-universal tree, but
in Appendix~\ref{sec:universal} we show that our technique implies a lower
bound of $\Omega(n^{2.185})$ on the size of a universal tree in the sense of Young et al.~\cite{YoungCW99},
which dramatically improves their lower bound of $\Omega(n\log n)$.

\begin{figure}[b]
\begin{center}
\begin{tabular}{l c c}
\hline
{\bf Trees} & {\bf Lower bound} & {\bf Upper bound} \\
\hline
Binary & $\Oh(n^{1.728})$ & $\Omega(n^{1.894})$ \bigstrut[t] \\
General &  $\Oh(n^{2.174})$ & $\Omega(n^{2.318})$ \bigstrut[b] \\
\hline
\end{tabular}
\end{center}
\caption{Summary of the new bounds on the size of minor-universal trees.}
\label{tbl:contribution}
\end{figure}

The drawback of our approach is that a labeling scheme obtained through a minor-universal tree is not necessarily effective, as
computing the label of the nearest common ancestor might require inspecting the (large)
minor-universal tree. However, in Section~\ref{sec:convert} we show that this is, in fact, not an issue at all:
any labeling scheme for nearest common ancestors based on a minor-universal tree
with labels of length $c\log n$ can be converted into a scheme with labels of length $c\log n+o(\log n)$
and constant query time. This further strengthens our claim that minor-universal
trees are the right approach for obtaining a clean bound on the size of the labels,
at least from the theoretical perspective (of course, in practice the $o(\log n)$
term might be very large).

\subsection{Our Techniques}

Our construction of a minor-universal tree for binary trees is recursive and based on
a generalization of the heavy path decomposition. In the standard heavy path decomposition
of a tree $T$, the top heavy path starts at the root and iteratively descends to the child corresponding to the largest
subtree. Depending on the version, this either stops at a leaf, or at a node corresponding
to a subtree of size less than $|T|/2$. After some
thought, a natural idea is to introduce a parameter $\alpha$ and stop after reaching
a node corresponding to a subtree of size less than $\alpha \cdot |T|$.
This is due to a certain imbalance between the subtrees rooted at the
children of the node where we stop and all subtrees hanging off the top heavy path.
Our minor-universal tree for binary trees on $n$ nodes consists of a long path to which the top heavy
path of any $T$ consisting of $n$ nodes can be mapped, and recursively defined smaller
minor-universal trees for binary trees of appropriate sizes attached to the nodes of the
long path. Hrubes et al.~\cite{HrubesWY10} also follow the same high-level idea, but
work with the path leading to a centroid node. In their construction, there is only one
minor-universal tree for binary trees on $2/3n$ nodes attached to the last node of the
path. We attach two of them: one for binary trees on $\alpha\cdot n$ nodes and one
for binary trees on $n/2$ nodes. We choose the minor-universal trees attached to the
other nodes of the long path using the same reasoning as Hrubes et al.~\cite{HrubesWY10}
(which is closely connected to designing an alphabetical code with codewords of given
lengths used in many labeling papers, see for example ~\cite{thorup2001compact}),
except that we can use a stronger bound on the total size of all subtrees hanging
off the top path and not attached to its last node than the one obtained from the
properties of a centroid node. Finally, we choose $\alpha$ as to optimize the whole
construction. Very similar reasoning, that is, designing a decomposition strategy 
by choosing the top heavy path with a cut-off parameter $\alpha$
and then choosing $\alpha$ as to minimize the total size, has been also used by Young
et al.~\cite{YoungCW99}, except that their definition of a universal tree is not the same as our
minor-universal tree and they do not explicitly phrase their reasoning in terms of a heavy path
decomposition, which makes it less clear.

To construct a minor-universal tree for general trees on $n$ nodes we need to somehow
deal with nodes of large degree. We observe that essentially the same construction works
if we use the following standard observation: if we sort the children of the root of $T$ by the
size of their subtrees, then the subtree rooted at the $i$-th child is of size at most $|T|/i$.

To show a lower bound on the size of a minor-universal tree for binary trees on $n$ leaves,
we also apply a recursive reasoning. The main idea is to consider $s$-caterpillars, which are
binary trees on $s$ leaves and $s-1$ inner nodes. For every node $u$ in the minor-universal
tree we find the largest $s$, such that an $s$-caterpillar can be mapped to the subtree rooted at $u$.
Then, we use the inductive assumption to argue that there must be many such nodes,
because we can take any binary tree on $\lfloor n/s\rfloor$ leaves and replace each of its
leaves by an $s$-caterpillar. For general trees, we consider slightly more complex
gadgets, and in both cases need some careful calculations.

To show that any labeling scheme based on a minor-universal tree can be converted
into a labeling scheme with roughly the same label length and constant decoding time,
we use a recursive tree decomposition similar to the one used by Thorup and Zwick~\cite{thorup2001compact},
and tabulate all possible queries for tiny trees.

\section{Preliminaries}

We consider rooted trees, and we think that every edge is directed from a parent
to its child. Unless mentioned otherwise, the trees are unordered, that is,
the relative order of the children is not important.
$\nca(u,v)$ denotes the nearest common ancestor of $u$ and $v$
in the tree. $\degree(u)$ denotes the degree (number of children) of $u$.
A tree is binary if every node has at most two children. For a rooted tree $T$,
$|T|$ denotes its size, that is, the number of nodes. In most cases, this will be denoted by
$n$. The whole subtree rooted at node $u\in T$ is denoted by $T^{u}$. 
If we say that $T'$ is a subtree of $T$, we mean that $T'=T^{u}$ for some $u\in T$,
and if we say that $T'$ is a subgraph of $T$, we mean that $T'$ can be obtained
from $T$ by removing edges and nodes.

If $s$ is a binary string, $|s|$ denotes its length, and we write $s<_{lex}t$ when $s$ is lexicographically
less than $t$. $\epsilon$ is the empty string.

Let $\mathcal{T}$ be a family of rooted trees. An NCA labeling scheme for $\mathcal{T}$ consists
of an encoder and a decoder. The encoder takes a tree $T\in \mathcal{T}$ and assigns a distinct
label (a binary string) $\nodelabel(u)$ to every node $u\in T$. The decoder receives labels $\nodelabel(u)$ and
$\nodelabel(v)$, such that $u,v\in T$ for some $T\in \mathcal{T}$, and should return
$\nodelabel(\nca(u,v))$. Note that the decoder is not aware of $T$ and only knows that $u$
and $v$ come from the same tree belonging to $\mathcal{T}$. We are interested in minimizing
the maximum length of a label, that is, $\max_{T\in \mathcal{T}}\max_{u\in T} |\ell(u)|$.

\section{NCA Labeling Schemes and Minor-Universal Trees}
\label{sec:upper}

We obtain an NCA labeling scheme for a class $\mathcal{T}$ of rooted trees by defining a minor-universal
tree $T$ for $\mathcal{T}$. $T$ should be a rooted tree with the property that, for any $T'\in\mathcal{T}$,
$T'$ is a topological minor of $T$, meaning that a subdivision of $T'$ is a subgraph of $T$. In other
words, it should be possible to map the nodes of $T'$ to the nodes of $T$ as to preserve the NCA
relationship: there should exist a mapping $f : T' \rightarrow T$ such that $f(\nca(u,v))=\nca(f(u),f(v))$
for any $u,v\in T'$.
We will define a minor-universal tree for trees on $n$ nodes, denoted by $U_{n}$,
and a minor-universal tree for binary trees on $n$ nodes, denoted by $B_{n}$. Note that $B_{n}$
does no have to be binary.

A minor-universal tree $U_{n}$ (or $B_{n}$) can be directly translated into an NCA labeling scheme as follows.
Take a rooted tree $T$ on $n$ nodes. By assumption, there exists a mapping $f: T \rightarrow U_{n}$ 
such that $f(\nca(u,v))=\nca(f(u),f(v))$ for any $u,v\in T$ (if there are multiple such mappings, we fix
one). Then, we define an NCA labeling scheme by choosing, for every $u\in T$, the label
$\ell(u)$ to be the (binary) identifier of $f(u)$ in $U_{n}$. The maximum length of a label in the obtained scheme is
$\lceil\log |U_{n}|\rceil$. In the remaining part of this section we thus focus
on defining small minor-universal trees $B_{n}$ and $U_{n}$.

\subsection{Binary Trees}

Before presenting a formal definition of $B_{n}$, we explain the intuition.

Consider a binary rooted tree $T$. We first explain the (standard) notion of heavy path
decomposition. For every non-leaf $u\in T$, we choose the edge leading to its
child $v$, such that $|T^{v}|$ is the largest (breaking ties arbitrarily). We call $v$ the heavy child of $u$.
This decomposes the nodes of $T$ into node-disjoint heavy paths. The topmost
node of a heavy path is called its head. All existing NCA labeling schemes are
based on some version of this notion and assigning variable length codes to
the roots of all subtrees hanging off the heavy path, so that larger subtrees receive shorter
codes. Then, the label of a node is obtained by concatenating the codes of
all of its light ancestors, or in other words ancestors that are heads of their heavy
paths.
There are multiple possibilities for how to define the codes (and how to concatenate
them while making sure that the output can be decoded). 
Constructing such a code is closely connected to the following lemma
used by Hrubes et al.~\cite{HrubesWY10} to define a minor-universal tree for
ordered binary trees. The lemma can be also extracted (with some effort, as it is
somewhat implicit) from the construction of Young et al.~\cite{YoungCW99}.

\begin{lemma}[see Lemma 8 of~\cite{HrubesWY10}]
\label{lem:order}
Let $a_{N}$ be a sequence recursively defined as follows: $a_{1}=(1)$,
and $a_{N}= a_{\lfloor N/2 \rfloor} \oplus (N) \oplus a_{ \lfloor N/2 \rfloor}$,
where $\oplus$ denotes concatenation. Then, for any sequence $b=(b(1),b(2),\ldots,b(k))$
consisting of positive integers summing up to at most $N$, there
exists a subsequence $a'$ of $a_{N}$ that dominates $b$, meaning that
the $i$-th element of $a'$ is at least as large as the $i$-th element of $b$,
for every $i=1,2,\ldots,k$.
\end{lemma}

%

The sequence defined in Lemma~\ref{lem:order} contains 1 copy of $N$,
2 copies of $\lfloor N/2\rfloor $, 4 copies of $\lfloor N/4\rfloor $, and so on.
In other words, there are $2^{i}$ copies of $\lfloor N/2^{i}\rfloor$, for every $i=0,1,\ldots,\lfloor \log N\rfloor $
there.

To present our construction of a minor-universal binary tree we need to modify the notion of heavy path decomposition.
Let $\alpha\in (1/2,1)$ be a parameter to be fixed later. We define $\alpha$-heavy path decomposition
as follows. Let $T$ be a rooted tree. We start at the root of $T$
and, as long as possible, keep descending to the (unique) child $v$ of the current node $u$, such
that $|T^{v}| \geq \alpha\cdot |T|$. This defines the top $\alpha$-heavy path.
Then, we recursively decompose every subtree hanging off the top $\alpha$-heavy
path.

Now we are ready to present our construction of the minor-universal binary tree $B_{n}$,
that immediately implies an improved nearest common ancestors labeling scheme 
for binary trees on $n$ nodes as explained in the introduction. $B_{0}$ is the empty tree
and $B_{1}$ consists of a single node.
For $n\geq 2$ the construction is recursive.
We invoke Lemma~\ref{lem:order} with $N=\lfloor(1-\alpha)n\rfloor$ to obtain a sequence
$a_{\lfloor(1-\alpha)n\rfloor}=(a(1),a(2),\ldots,a(k))$. Then, $B_{n}$ consists of a path
$u_{1}-u_{2}-\ldots -u_{k+1}$. We attach a copy of $B_{a(i)-1}$ to every
$u_{i}$, for $i=1,2,\ldots,k$. Additionally, we attach a copy of
$B_{\lfloor\alpha\cdot n\rfloor}$ and $B_{\lfloor(n-1)/2\rfloor}$ to $u_{k+1}$. 
See Figure~\ref{fig:binary}.
Note that $a(i)-1<n$, $\lfloor \alpha \cdot n \rfloor < n$ for $\alpha < 1$, and
$\lfloor (n-1)/2 \rfloor < n$, so this is indeed a valid recursive definition.
We claim that $B_{n}$ is a minor-universal tree for all binary
trees on $n$ nodes.

\begin{figure}[t]
\begin{center}
\includegraphics[scale=0.5]{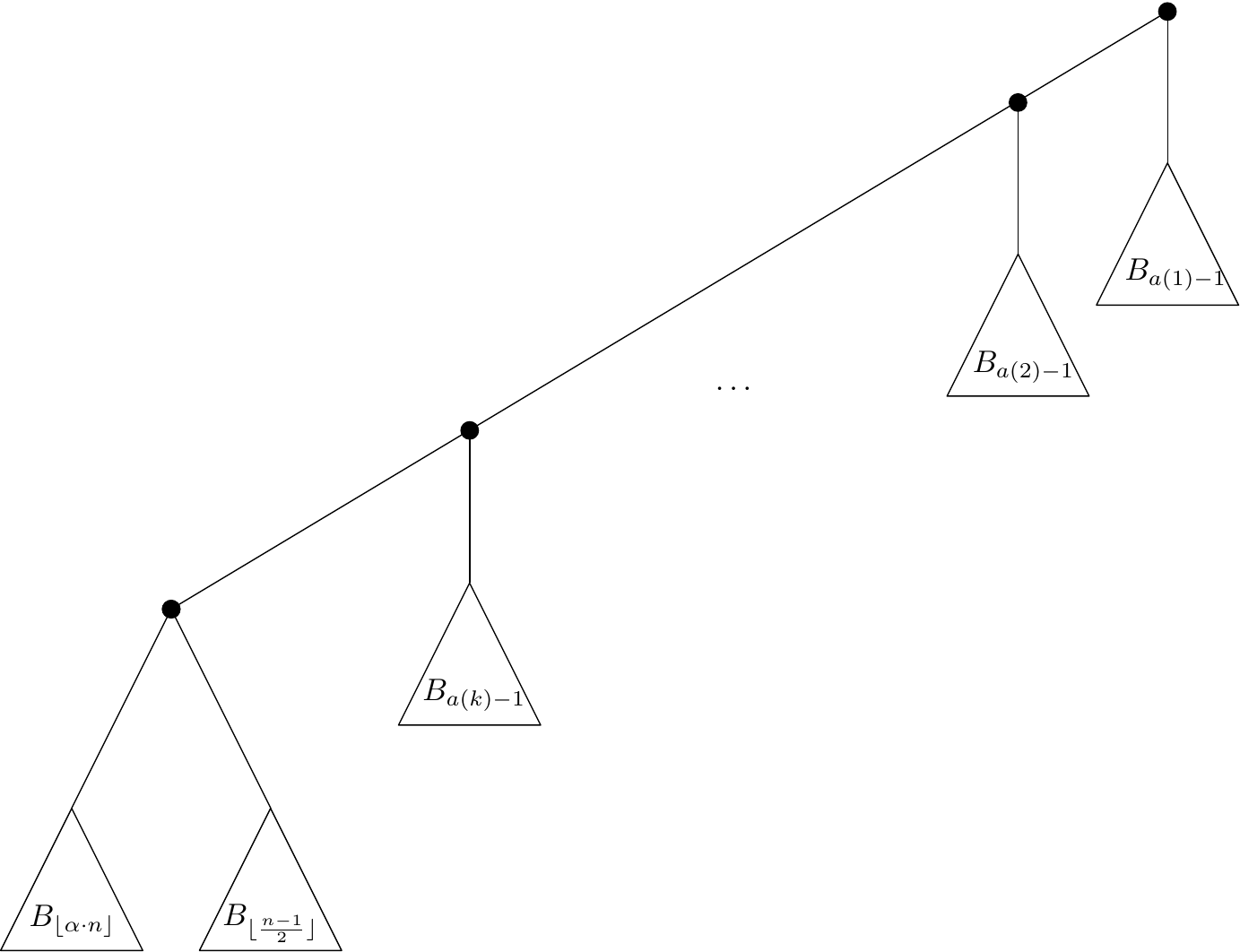}
\end{center}
\caption{A schematic illustration for the recursive construction of $B_{n}$.}
\label{fig:binary}
\end{figure}

\begin{lemma}
\label{lem:binary}
For any binary tree $T$ on $n$ nodes, $B_{n}$ contains a subgraph isomorphic
to a subdivision of $T$.
\end{lemma}

\begin{proof}
We prove the lemma by induction on $n$.

Consider a binary tree $T$ on $n\geq 2$ nodes and let $v_{1} - v_{2} - \ldots - v_{s}$ be the
path starting at the root in the $\alpha$-heavy path decomposition of $T$. Then,
$|T^{v_{s}}| \geq \alpha \cdot n$, but for every child $u$ of $v_{s}$ we have that
$|T^{u}| < \alpha \cdot n$. Consequently, the total size of all subtrees hanging off the
path and attached to $v_{1},v_{2},\ldots,v_{s-1}$, increased by $s-1$, is at most $(1-\alpha)n$.
Also, denoting by $u_{1}$ and $u_{2}$ the children of $v_{s}$ and ordering them so that $|T^{u_{1}}| \geq |T^{u_{2}}|$,
we have $|T^{u_{1}}| < \alpha \cdot n$ and $|T^{u_{2}}| \leq (n-1)/2$ (we assume that $v_{s}$ has
two children, otherwise we can think that the missing children are of size 0). Then,
by the inductive assumption,
a subdivision of $T^{u_{1}}$ is a subgraph of $B_{\lfloor\alpha\cdot n\rfloor}$,
and a subdivision of $T^{u_{2}}$ is a subgraph of $B_{\lfloor(n-1)/2\rfloor}$.
Further, denoting by $b(i)-1$ the size of the subtree hanging off the path and attached to $v_{i}$,
we have $\sum_{i=1}^{s-1} b(i) \leq m$, where $m=\lfloor (1-\alpha)n\rfloor$, and every $b(i)$ is positive, so
$b$ is dominated by a subsequence of $a_{m}=(a(1),a(2),\ldots,a(k))$. This means that we can
find indices $1\leq j(1) < j(2) < \ldots < j(s-1)\leq k$, such that $b(i) \leq a(j(i))$, for every
$i=1,2,\ldots,s-1$. But then a subdivision of the subtree hanging off the path attached
to $v_{i}$ is a subgraph of $B_{a(j(i))-1}$ attached to $u_{j(i)}$ in $B_{n}$.
Together, all these observations imply that a subdivision of the
whole $T$ is a subgraph of $B_{n}$.
\end{proof}

Finally, we analyze the size of $B_{n}$. Because $B_{n}$ consists of a copy of
$B_{\lfloor\alpha\cdot n\rfloor}$, a copy of $B_{\lfloor (n-1)/2\rfloor}$, and $2^{i}$ copies of $B_{\lfloor (1-\alpha)n/2^{i}\rfloor-1}$
attached to the path, for every $i=0,1,\ldots,\lfloor\log(1-\alpha)n\rfloor$, we have
the following recurrence:
\begin{align*}
|B_{n}| &= 1 + |B_{\lfloor\alpha\cdot n\rfloor }| + |B_{\lfloor (n-1)/2\rfloor}| + \sum_{i=0}^{\lfloor\log(1-\alpha)n\rfloor} 2^{i} \cdot (1+|B_{\lfloor (1-\alpha)n/2^{i}\rfloor-1}|) 
\end{align*}
We want to inductively prove that $|B_{n}|\leq n^{c}$, for some (hopefully small) constant $c>1$. To this end,
we introduce a function $b(x)=x^{c}$ that is defined for any real $x$, and try to show that
$|B_{n}|\leq b(n)$ by induction on $n$.
Using the inductive assumption, $1+|B_{m}| \leq b(m+1)$ holds for any $m<n$ by applying the Bernoulli's inequality
and checking $m=0$ separately. We would also like to use $1+|B_{\lfloor (n-1)/2\rfloor}| \leq b(n/2)$
for $n\geq 2$, which requires additionally verifying that $1+k^{c} \leq (k+1/2)^{c}$ for
any $k\geq 1$. For $c>1.71$, this holds for $k\geq 2$ by the Bernoulli's inequality and can
be checked for $k=1$ separately. These inequalities allow us to upper bound $|B_{n}|$ as follows:
\begin{align*}
|B_{n}| &\leq b(\alpha\cdot n) + b(n/2) + \sum_{i\geq 0} 2^{i} \cdot b((1-\alpha)n/2^{i})
\end{align*}
To conclude that indeed $|B_{n}| \leq b(n)$, it suffices that the following inequality holds:
\begin{align*}
\alpha^{c} + (1/2)^{c} + \sum_{i\geq 0} 2^{i} ((1-\alpha)/2^{i})^{c} &\leq 1 \\
\alpha^{c} + (1/2)^{c} + (1-\alpha)^{c} \sum_{i\geq 0} (1/2^{c-1})^{i} &\leq 1 \\
\alpha^{c} + (1-\alpha)^{c} \cdot 2^{c-1} / (2^{c-1}-1) &\leq 1-(1/2)^{c}
\end{align*}
Minimizing $f(x)=x^{c}+(1-x)^{c} \cdot 2^{c-1} / (2^{c-1}-1)$ we obtain
$x=A/(1+A)$, where $A=(2^{c-1}/(2^{c-1}-1))^{1/(c-1)}$. Thus, it is enough
that $(A/(1+A))^{c}+(1/A)^{c}\cdot 2^{c-1} / (2^{c-1}-1) \leq 1-(1/2)^{c}$.
This can be solved numerically for the smallest possible $c$ and verified to
hold for $c=1.894$ by choosing $A=2.372$ and $\alpha=0.704$.

\subsection{General Trees}

To generalize the construction to non-binary trees, we use the same notion of
$\alpha$-heavy path decomposition. Again, we invoke
Lemma~\ref{lem:order} with $N=(1-\alpha)n$ to obtain a sequence
$a_{\lfloor (1-\alpha)n\rfloor}=(a(1),a(2),\ldots,a(k))$. $U_{n}$ consists of a path
$u_{1}-u_{2}-\ldots -u_{k+1}$. For every $i=1,2,\ldots,k$,
we attach a copy of $U_{a(i)-1}$ and, for every $j\geq 2$,
a copy of $U_{\lfloor a(i)/j\rfloor}$ to $u_{i}$.
Additionally, we attach a copy of
$U_{\lfloor \alpha\cdot n\rfloor}$ to $u_{k+1}$ and also, for every $j\geq 2$,
a copy of $U_{\lfloor (n-1)/j\rfloor}$. See Figure~\ref{fig:general}.
We claim that $U_{n}$ is indeed a minor-universal tree for all 
trees on $n$ nodes.

\begin{figure}[t]
\begin{center}
\includegraphics[scale=0.5]{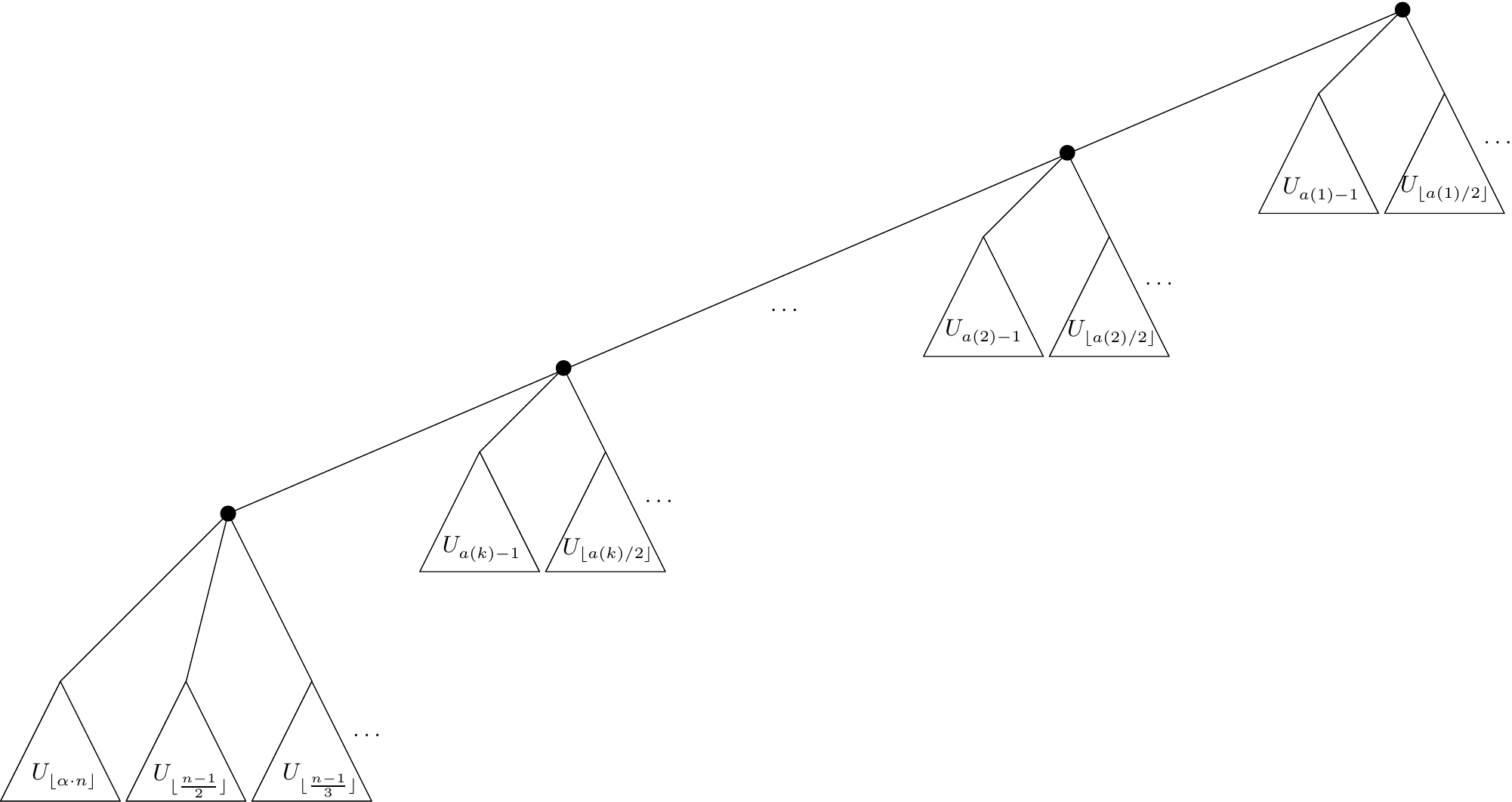}
\end{center}
\caption{A schematic illustration for the recursive construction of $U_{n}$.}
\label{fig:general}
\end{figure}

\begin{lemma}
\label{lem:general}
For any tree $T$ on $n$ nodes, $U_{n}$ contains a subgraph isomorphic
to a subdivision of $T$.
\end{lemma}

\begin{proof}
The proof is very similar to the proof of Lemma~\ref{lem:binary}.

Consider a  tree $T$ on $n$ nodes and let $v_{1} - v_{2} - \ldots - v_{s}$ be the
path starting at the root in the $\alpha$-heavy path decomposition of $T$. Again,
the total size of all subtrees hanging off the path and attached to $v_{1},v_{2},\ldots,v_{s-1}$,
increased by $s-1$, is less than $(1-\alpha)n$, and if we denote by $u_{1},u_{2},u_{3},\ldots$ the children of $v_{s}$
and order them so that $|T^{u_{1}}| \geq |T^{u_{2}}| \geq |T^{u_{3}}| \geq \ldots$
then $|T^{u_{1}}| < \alpha \cdot n$ and $|T^{u_{j}}| \leq (n-1)/j$ for every $j\geq 2$.
Then, a subdivision of $T^{u_{1}}$ is a subgraph of $U_{\lfloor\alpha\cdot n\rfloor}$,
and a subdivision of $T^{u_{j}}$ is a subgraph of $U_{\lfloor(n-1)/j\rfloor}$,
for every $j=2,3,\ldots$.
Denoting by $b(i)-1$ the total size of all subtrees hanging off the path and attached
to $v_{i}$, we can find indices $1\leq j(1) < j(2) < \ldots < j(s-1)$, such that $b(i) \leq a(j(i))$, for every
$i=1,2,\ldots,s-1$, where $a_{\lfloor (1-\alpha)n\rfloor}=(a(1),a(2),\ldots,a(k))$.
Let $v(i,1),v(i,2),\ldots$ be the children of
$v_{i}$ ordered so that $|T^{v(i,1)}| \geq |T^{v(i,2)}| \geq \ldots$.
Then a subdivision of $T^{v(i,1)}$ is a subgraph of $U_{a(j(i))-1}$ attached
to $u_{j(i)}$ in $U_{n}$ and, for every $k\geq 2$,
a subdivision of $T^{v(i,k)}$ is a subgraph of $U_{\lfloor a(j(i))/k\rfloor}$
attached to the same $u_{j(i)}$ in $U_{n}$.
This all imply that a subdivision of the whole $T$ is a subgraph of $U_{n}$.
\end{proof}

To analyze the size of $U_{n}$, observe that $|U_{n}|$ can be bounded by
\begin{align*}
1 + |U_{\lfloor\alpha\cdot n\rfloor }| + \sum_{i=2}^{n-1}|U_{\lfloor (n-1)/i\rfloor}| + \sum_{i=0}^{\lfloor\log (1-\alpha)n\rfloor} 2^{i} \cdot (1+|U_{\lfloor (1-\alpha)n/2^{i}\rfloor-1}|+\sum_{j=2}^{\lfloor (1-\alpha)n/2^{i}\rfloor}|U_{\lfloor (1-\alpha)n/(2^{i}\cdot j)\rfloor}|) 
\end{align*}
We want to inductively prove that $|U_{n}| \leq s(n)$, where $s(n)=n^{c}$, for some constant $c>1$.
By the same reasoning as the one used to bound $|B_{n}|$:
\begin{align*}
|U_{n}| &\leq s(\alpha\cdot n) + \sum_{i\geq 2}s(n/i) + \sum_{i\geq 0} 2^{i} \sum_{j\geq 1} s((1-\alpha)n/(2^{i}\cdot j))
\end{align*}
For the inductive step to hold, it suffices that:
\begin{align*}
\alpha^{c} + \sum_{i\geq 2}(1/i)^{c} + \sum_{i\geq 0} 2^{i} \sum_{j\geq 1}((1-\alpha)/(2^{i}\cdot j))^{c} &\leq 1 \\
\alpha^{c} + \sum_{i\geq 1}(1/i)^{c} + \sum_{i\geq 0} 2^{i} ((1-\alpha)/2^{i})^{c}\cdot \zeta(c) &\leq 2 \\
\alpha^{c} +(1-\alpha)^{c} \cdot \zeta(c) 2^{c-1} / (2^{c-1}-1) &\leq 2-\zeta(c)
\end{align*}
where $\zeta(c)=\sum_{i=1}^{\infty} 1/i^{c}$ is the standard Riemann zeta function. 
Recall that our goal is to make $c$ as small as possible, and we can adjust $\alpha$. By
approximating $\zeta(c)$, we can verify that, after choosing $\alpha=0.659$,
the above inequality holds for $c=2.318$.

\section{Lower Bound for Minor-Universal Trees}
\label{sec:lower}

In this section, we develop a lower bound on the number of nodes in a minor-universal
tree for binary trees on $n$ nodes, and a minor-universal tree for general trees on $n$
nodes. In both cases, it is convenient to lower bound the number of leaves in a tree,
that contains as a subgraph a subdivision of any binary tree (or a general tree) $T$,
such that $T$ contains $n$ leaves and no degree-1 nodes. This is denoted by
$b(n)$ and $u(n)$, respectively. Because we do not allow degree-1 nodes in $T$,
it has at most $2n-1$ nodes, thus $b(n)$ is a lower bound on the size of a minor-universal tree for binary trees on
$2n-1$ nodes, and similarly $u(n)$ is a lower bound on the size of a minor-universal
tree for general trees on $2n-1$ nodes.

\subsection{Binary Trees}

We want to obtain a lower bound on the number of leaves $b(n)$ in a tree,
that contains as a subgraph a subdivision of any binary tree $T$ on $n$ leaves
and no degree-1 nodes.

\begin{lemma}
\label{lem:lowerbinary}
$b(n) \geq 1+\sum_{s\geq 2} b(\lfloor n/s \rfloor).$
\end{lemma}

\begin{proof}
For any $s\geq 2$, we define an $s$-caterpillar to be a binary tree on $s$ leaves
and $s-1$ inner nodes creating a path. Consider a tree $T$, that
contains as a subgraph a subdivision of any binary tree on $n$ leaves and no
degree-1 nodes.
For a node $v\in T$, let $s(v)$ be the largest $s$, such that $T^{v}$
contains a subdivision of an $s$-caterpillar as a subgraph. We say
that such $v$ is on level $s(v)$. We observe that $s(v)$ has the following properties:
\begin{enumerate}
\item For every child $u$ of $v$, $s(u) \leq s(v)$.
\item If the degree of $v$ is 1 then, for the unique child $u$ of $v$, $s(u)=s(v)$.
\item If the degree of $v$ is 2 then, for some child $u$ of $v$, $s(u)=s(v)-1$.
\end{enumerate}

Choose a parameter $s\geq 2$ and consider any binary tree on
$\lfloor n/s \rfloor$ leaves and no degree-1 nodes. By replacing all of its leaves by
$s$-caterpillars we obtain a binary tree on at most $n$ leaves and still
no degree-1 nodes. A subdivision of this new binary tree must be a subgraph of $T$. The leaves
of the original binary tree must be mapped to nodes on level at
least $s$ in $T$.
Thus, by removing all nodes on level smaller than $s$ from $T$ we obtain a tree $T'$
that contains as a subgraph a subdivision of
any binary tree on $\lfloor n/s \rfloor$ leaves and no degree-1 nodes, and
so there are at least $b(\lfloor n/s\rfloor)$ leaves in $T'$.
By the properties of $s(u)$, a leaf of $T'$ corresponds to a node $u\in T$
on level $s$ (as otherwise $u$ has a child on level at least $s$), and
furthermore the degree of $u$ must be at least 2 (as otherwise the only
child of $u$ is on the same level).
Because the level of every $u\in T$ is unambiguously defined,
the total number of degree-2 nodes in $T$ is at
least:
\[ \sum_{s\geq 2} b(\lfloor n/s\rfloor) \]
To complete the proof, observe that in any tree the number of leaves
is larger than the number of degree-2 nodes.
\end{proof}

We want to extract an explicit lower bound on $b(n)$ from Lemma~\ref{lem:lowerbinary}.

\begin{theorem}
\label{thm:binary}
For any $c>1$ such that $\zeta(c) > 2$ we have $b(n) = \Omega(n^{c})$.
\end{theorem}

\begin{proof}
We assume that $\zeta(c) > 2$, so $\sum_{s=2}^{\infty} (1/s)^{c} > 1$. Then there exists
$\epsilon>0$ and $t$, such that $\sum_{s=2}^{t} (1/s)^{c} = 1+\epsilon$.

Now consider a function $f(x)=x^{c}$.
We claim that there exists $x_{0}$, such that for all $x \geq x_{0}$ we have
$f(\lfloor x \rfloor) \geq  f(x-1) \geq f(x)/(1+\epsilon)$. This is because of the following
transformations:
\begin{align*}
f(x-1) &\geq f(x) / (1+\epsilon) \\
(x-1)^{c} &\geq x^{c} / (1+\epsilon) \\
1-1/x &\geq 1/(1+\epsilon)^{1/c} 
\end{align*}
where the right side is smaller than 1, so the inequality holds for any sufficiently large
$x$.

We are ready to show that $b(n) \geq a \cdot n^{c}$ for some constant $a$. We proceed by induction on $n$.
By Lemma~\ref{lem:lowerbinary}, we know that $b(n) \geq \sum_{s\geq 2} b(\lfloor n/s \rfloor)$.
By adjusting $a$, it is enough to show that, for sufficiently large values of $n$,
$b(N) \geq a\cdot N^{c}$ holding for all $N < n$ implies $b(n) \geq a\cdot n^{c}$.
We lower bound $b(n)$ as follows:
\begin{align*}
b(n) &> \sum_{s\geq 2} b(\lfloor n/s \rfloor )  
\geq \sum_{s\geq 2} a\cdot (\lfloor n/s\rfloor)^{c} 
\geq a\sum_{s=2}^{n/x_{0}} (\lfloor n/s\rfloor)^{c} 
 \geq a\cdot n^{c} \cdot \sum_{s=2}^{n/x_{0}} (1/s)^{c} / (1+\epsilon)
\end{align*}
where in the last inequality we used that, as explained in the previous paragraph,
$(\lfloor n/s\rfloor)^{c} \geq (n/s)/(1+\epsilon)$ for $n/s \geq x_{0}$.
By restricting $n$ to be so large that $n/x_{0} \geq t$, i.e., $n \geq t\cdot x_{0}$,
we further lower bound $b(n)$ as follows:
\begin{align*}
b(n) &\geq a\cdot n^{c} \cdot \sum_{s=2}^{t} (1/s)^{c} / (1+\epsilon) 
 \geq a\cdot n^{c} \cdot (1+\epsilon)/ (1+\epsilon) 
 = a \cdot n^{c}\qedhere
\end{align*}
\end{proof}

To apply Theorem~\ref{thm:binary}, we verify with numerical calculation that
$\zeta(1.728) > 2$, and so $b(n) = \Omega(n^{1.728})$.

\subsection{General Trees}

Now we move to general trees. We want to lower bound the number of leaves $u(n)$
in a tree, that contains as a subgraph a subdivision of any tree on $n$ leaves and no
degree-1 nodes.

We start with lower bounding the number of nodes of degree at
least $d$ in such a tree, denoted $u_{\geq}(n,d)$. Similarly, $u(n,d)$ denotes the number
of nodes of degree exactly $d$.

\begin{lemma}
\label{lem:lowergeneral}
For any $d\geq 2$, we have $u_{\geq}(n,d) \geq \sum_{s\geq 2}u(\lfloor n/((s-1)(d-1)+1) \rfloor)$.
\end{lemma}

\begin{proof}
Fix $d\geq 2$. For any $s\geq 2$, we define an $(s,d)$-caterpillar to consist of
path of length $s-1$, where we connect $d-1$ leaves to every node except for the
last, where we connect $d$ leaves. The total number of leaves in an $(s,d)$-caterpillar
is hence $(s-1)(d-1)+1$.
Consider a tree $T$, that contains as a subgraph a subdivision of any tree on $n$
nodes and no degree-1 nodes.
For any node $v\in T$, let $s(v)$ be the largest $s$, such that
$T^{v}$ contains a subdivision of an $(s,d)$-caterpillar as a subgraph. This is
a direct generalization of the definition used in the proof of Lemma~\ref{lem:lowerbinary},
and so similar properties hold:
\begin{enumerate}
\item For every child $u$ of $v$, $s(u) \leq s(v)$.
\item If the degree of $v$ is less than $d$ then, for some child $u$ of $v$, $s(u)=s(v)$.
\item If the degree of $v$ is at least $d$ then, for some child $u$ of $v$, $s(u)=s(v)-1$.
\end{enumerate}

Choose any $s\geq 2$ and consider a tree on $\lfloor n/((s-1)(d-1)+1)\rfloor$ leaves.
By replacing all of its leaves by $(s,d)$-caterpillars, we obtain a tree on at most
$n$ leaves, so subdivision of this new tree must be a subgraph of $T$. The
leaves of the original tree must be mapped to nodes on level at least $s$
in $T$, and by the same reasoning as in the proof of Lemma~\ref{lem:lowerbinary}
this implies that there are at least $u(\lfloor n/((s-1)(d-1)+1)\rfloor)$ nodes
on level $s$ and of degree at least $d$ in $T$, making the total number of nodes
of degree $d$ or more at least:
\[ \sum_{s\geq 2} u(\lfloor n/((s-1)(d-1)+1) \rfloor) \qedhere\]
\end{proof}

\begin{proposition}
\label{prop:degrees}
Let $n(d)$ denote the number of nodes of degree $d$, then the number of leaves
is $1+\sum_{d\geq 1} n(d)\cdot (d-1)$.
\end{proposition}

\begin{proof}
We apply induction on the size of the tree. If the tree consists of only one
node, then the claim holds. Assume that the root is of degree $a\geq 1$.
Then, by applying the inductive assumption on every subtree attached to
the root and denoting by $n'(d)$ the number of non-root nodes of degree
$d$, we obtain that the number of leaves is $a+\sum_{d} n'(d)\cdot (d-1)$.
Finally, $n(d)$ is $n'(d)$ if $d\neq a$, and $n(a)=n'(a)+1$, so the number of
leaves is in fact:
\[ a+n'(a)\cdot (a-1) + \sum_{d\neq a} n'(d)\cdot (d-1) = 1+\sum_{d}n(d)\cdot (d-1) \qedhere \]
\end{proof}

By combining Proposition~\ref{prop:degrees} with $u_{\geq}(n,d) = u(n,d)+u_{\geq}(n,d+1)$ and telescoping,
we obtain that the number of leaves is at least:
\[ 1+\sum_{d\geq 2} u(n,d)\cdot (d-1)  = 1 + \sum_{d\geq 2} (u_{\geq}(n,d)-u_{\geq}(n,d+1)) \cdot (d-1) = 1 + \sum_{d\geq 2} u_{\geq}(n,d)\]
Finally, by substituting Lemma~\ref{lem:lowergeneral} we obtain:
\[ u(n) \geq 1+\sum_{d\geq 2} \sum_{s\geq 2} u(\lfloor n/((s-1)(d-1)+1)\rfloor) \]

\begin{theorem}
\label{thm:general}
For any $c>1$ such that $\sum_{x,y\geq 1} 1/(x\cdot y+1)^{c} > 1$ we have
$u(n) = \Omega(n^{c})$.
\end{theorem}

\begin{proof}
We extend the proof of Theorem~\ref{thm:binary}. From $\sum_{x,y\geq 1} 1/(x\cdot y+1)^{c} > 1$
we obtain that there exists $\epsilon>0$ and $t$, such that $\sum^{t}_{x=1}\sum^{t}_{y=1} 1/(x\cdot y+1)^{c} = 1+\epsilon$.

We know that $u(n) \geq \sum_{s,d\geq 2} u(\lfloor n/((s-1)(d-1)+1) \rfloor)$. As in the proof
of Theorem~\ref{thm:binary}, we only need to show that $u(n) \geq a \cdot n^{c}$ for sufficiently large $n$.
We lower bound $u(n)$:
\begin{align*}
u(n) &\geq \sum_{s,d\geq 2} u(\lfloor n/((s-1)(d-1)+1) \rfloor) \\
&\geq \sum^{t}_{x=1}\sum^{t}_{y=1} a\cdot (\lfloor n/(x\cdot y+1)\rfloor)^{c} \\
& \geq a\cdot n^{c}\cdot \sum_{x\geq 1} \sum_{y=1}^{n/(x_{0}\cdot x)} (1/(x\cdot y+1))^{c} / (1+\epsilon)
\end{align*}
We choose $n$ so large that $n/(x_{0}\cdot t) \geq t$ and further lower bound $u(n)$:
\begin{align*}
u(n) &\geq a\cdot n^{c}\cdot \sum_{x=1}^{t} \sum_{y=1}^{t} (1/(x\cdot y+1))^{c} / (1+\epsilon) \\
&\geq a\cdot n^{c} \cdot (1+\epsilon) / (1+\epsilon) \\
&= a\cdot n^{c} \qedhere
\end{align*}
\end{proof}

We verify with numerical calculations that $\sum_{x,y\geq 1} 1/(x\cdot y+1)^{2.174} > 1$ by computing
the sum
$\sum^{1000}_{x=1}\sum^{1000}_{y=1} 1/(x\cdot y+1)^{2.174}$ and conclude that $u(n) =\Omega(n^{2.174})$.

\section{Complexity of the Decoding}
\label{sec:convert}

In this section we present a generic transformation, that converts our existential results into
labeling schemes with constant query time in the word-RAM model with word size $\Omega(\log n)$.
Our goal is to show the following statement for any class $\mathcal{T}$ of rooted trees closed
under taking topological minors: if, for any $n$, there exists a minor-universal tree of size
$n^{c}$, then there exists a labeling scheme with labels consisting of $c\log n+o(\log n)$
bits, such that given the labels of two nodes we can compute the label of their nearest common
ancestor in constant time. We focus on general trees, and leave verifying that the same method
works for any such $\mathcal{T}$ to the reader.

Before proceeding with the main part of the proof, we need a more refined method of converting
a minor-universal tree into a labeling scheme for nearest common ancestors. Intuitively, we
would like some nodes to receive shorter labels. This can be enforced with the following lemma.

\begin{lemma}
\label{lem:weighted}
For any tree $T$, it is possible to assign a distinct label $\ell(u)$ to every node $u\in T$,
such that $|\ell(u)| \leq 2+\log ( |T| / (1+\degree(u)))$.
\end{lemma}

\begin{proof}
We partition the nodes of $T$ into classes. For every $k=0,1,\ldots,\lfloor\log(|T|+1)\rfloor$,
the $k$-th class contains all nodes with degree from $[2^{k}-1,2^{k+1}-1)$.
Observe that the sum of degrees of all nodes of $T$ is $|T|-1$, and consequently
the $k$-th class consists of at most $(2|T|-1)/2^{k}$ nodes. Thus,
we can assign a distinct binary code of length $\lceil\log((2|T|-1)/2^{k})\rceil\leq 2+\log(|T| / (1+\degree(u)))$
as the label of every node $u$ in the $k$-th class.
The length of the code uniquely determines $k$ so the labels are indeed distinct.
\end{proof}

Let $b$ be a parameter to be fixed later. To define the labels of all nodes of a tree $T$, we
recursively decompose it into smaller trees as follows, similarly to~\cite{thorup2001compact}.
First, we call a node $u$ such that $|T^{u}|\geq n/b$
big, and small otherwise. Let $T'$ be the subgraph of $T$ consisting of all big nodes (notice that
if $u$ is big, then so is its parent). Then there are at most $b$ leaves in $T'$, as each of
them corresponds to a disjoint subtree of size at least $n/b$. Therefore, there are less than $b$
branching nodes of $T'$.  We call all leaves, all branching nodes, all big children of branching nodes
and, if the root of $T'$ has exactly one big child, also the root and its only big child, interesting.
The total number of nodes designated as interesting
so far is $\Oh(b)$. We additionally attach virtual interesting nodes to some interesting
nodes as follows. For a big node $u$, $\weight(u)$ is defined as the total size
of all small subtrees attached to it.
If $u$ is the root, a leaf, or a branching node of $T'$, then we attach $\lceil \weight(u)/(n/b)\rceil$
virtual interesting nodes as its children. If $u$ is a big child of an interesting node,
then there is a unique path $v=v_{0}-v_{1}-\ldots v_{s}=u$ from a leaf or a branching node of $T'$
to $u$ that do not contain any other interesting nodes. We attach $\lceil \sum_{i=1}^{s}\weight(v_{i})/(n/b)\rceil$ virtual
interesting nodes as children of $u$.
We denote by $T^{c}$ the tree induced by all interesting nodes (including the virtual ones),
meaning that its nodes are all interesting nodes and the parent of a non-root interesting node
is its nearest interesting ancestor in $T$.
Because the sum of $\weight(u)$ over all big nodes $u$ is at most $n$, the total number
of virtual interesting nodes is $\Oh(b)$.
Thus, the total number of nodes in $T^{c}$ is $\Oh(b)=a\cdot b$.
Further, any interesting node has at most one big non-interesting child.
See Figure~\ref{fig:partition} for a schematic illustration of such a partition. Every
subtree rooted at a small node such that its parent is big is then decomposed recursively
using the same parameter $b$. Observe that the depth of the recursion is $\log_{b}n = \log n / \log b$.

\begin{figure}[t]
\begin{center}
\includegraphics[scale=0.5]{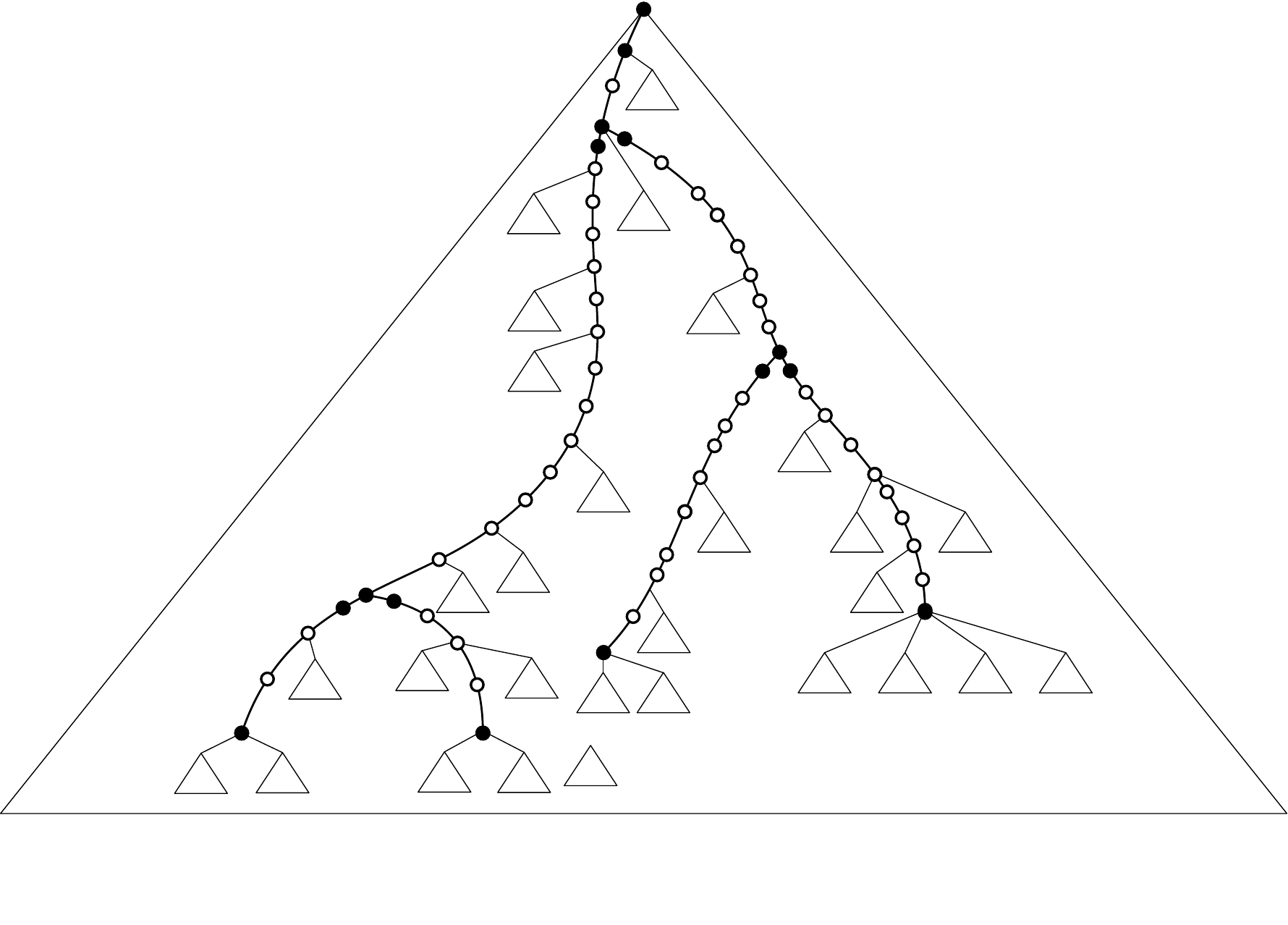}
\end{center}
\caption{A schematic illustration of a partition by choosing $\Oh(b)$ interesting nodes. Circles represent big nodes, and filled circles
represent interesting big nodes.}
\label{fig:partition}
\end{figure}

We are ready to define a query-efficient labeling scheme. The label of every node consists of
$\Oh(\log n / \log b)$ variable-length nonempty fields $f_{1},f_{2},\ldots,f_{s}$ and some shared
auxiliary information which will be explained later.
The fields are simply concatenated together, and hence also need to separately store
$|f_{1}|,|f_{2}|,|f_{3}|,\ldots,|f_{s-1}|$. This is done by extending a standard construction
as explained in Appendix~\ref{sec:proofs}.

\begin{restatable}{lemma}{rankselectlemma}
\label{lem:rank/select}
Any set of at most $s$ integers from $[1,M]$, such that $M=\Oh(\log n)$, can be encoded with
$\Oh(s\cdot \max\{1,\log\frac{M}{s}\})$ bits, so that we can implement the following operations in constant time:
\begin{enumerate}
\item extract the $k^\text{th}$ integer,
\item find the successor of a given $x$,
\item construct the encoding of a new set consisting of the smallest $k$ integers.
\end{enumerate}
The encoding depends only on the stored set (and the values of $s$ and $M$) and not on
how it was obtained.
\end{restatable}

Lemma~\ref{lem:rank/select} is applied to the set containing all numbers of the form
$\sum_{j=1}^{i}|f_{j}|$, for $i=1,2,\ldots,s-1$. Then, given a position in the concatenation
we can determine in constant time which field does it belong to, or find the first position
corresponding to a given field. We can also truncate the concatenation to contain only
$f_{1},f_{2},\ldots,f_{i}$ in constant time.

The $i$-th step of the recursive decomposition corresponds to three fields
$f_{3i-2},f_{3i-1},f_{3i}$, except that the last step corresponds to between one
and two fields. Below we describe how the fields corresponding to a single step 
are defined.

Consider a node $u\in T$ and let $u'$ be its nearest interesting ancestor. The first field
is the label of $u'$ obtained from by applying Lemma~\ref{lem:weighted} on $T^{c}$.
If $u=u'$ then we are done. Otherwise, $u\neq u'$, and we have two possibilities.
If $u'$ is a leaf or a branching node of $T'$, the second field contains a single 0.
Otherwise, $u'$ is a big child of a branching node or the root of $T'$, and the nearest
big ancestor of $u$ is some $v_{j}$ on a path $v_{0}-v_{1}-\ldots-v_{s}=u'$,
where $j\in\{1,2,\ldots,s\}$ and $v_{0}$ is a leaf or a branching node of $T'$.
In such case, we assign binary codes to all nodes $v_{1},v_{2},\ldots,v_{s}$ and
choose the second field to contain the code $\ell_{j}$ assigned to node $v_{j}$.
The codes should have the property
that $|\ell_{j}| \leq 1+\log((\sum_{i=1}^{s}\weight(v_{i}))/\weight(v_{j}))$, and
furthermore $\ell_{1}<_{lex} \ell_{2}<_{lex} \ldots <_{lex} \ell_{s}$. Such codes
can be obtained by the following standard lemma, that essentially follows
by the reasoning from Lemma~\ref{lem:order}.
This is almost identical to Lemma 2.4 in~\cite{thorup2001compact} (or Lemma 4.7
of~\cite{AHL14}), except that we prefer the standard lexicographical order.

\begin{lemma}
Given positive integers $b_{1},b_{2},\ldots,b_{m}$ and denoting $B=\sum_{i=1}^{m}b_{i}$,
we can find nonempty binary strings $s_{1}<_{lex} s_{2}<_{lex}\ldots <_{lex} s_{m}$, such that
$|s_{i}| \leq 1+\log(B/b_{i})$.
\end{lemma}

\begin{proof}
We choose the largest $j$, such that $\sum_{i<j}b_{j} \leq \lfloor B/2\rfloor$. We set
$s_{j+1}=1$. We recursively define the binary strings for $b_{1},b_{2},\ldots,b_{j}$
and prepend $0$ to each of them. Then, we recursively define the binary strings for $b_{j+2},\ldots,b_{m}$
and prepend $1$ to each of them. To verify that $s_{i}\leq 1+\log (B/b_{i})$ holds, observe
that the sum decreases by a factor of at least 2 in every recursive call, and the length of the binary
strings increases by 1.
\end{proof}

Let $u''$ be the nearest big ancestor of $u$. If $u=u''$ then we are done.
Otherwise, let $v_{1},v_{2},\ldots,v_{d}$ be all the small children of $u''$. We order
them so that $|T^{v_{i}}| \geq |T^{v_{i+1}}|$ for $i=1,2,\ldots,d-1$. Then,
$u$ belongs to the subtree $T^{v_{k}}$, for some $k\in \{1,2,\ldots,d\}$. 
The third and final field is simply the binary encoding of $k$ consisting of
$1+\lfloor\log k\rfloor\leq 1+\log k$ bits.
This completes the description of the fields appended to the label in a single step of the recursion.

In every step of the recursion, the size of the current tree decreases by at least
a factor of $b$. If we could guarantee that the total length of all fields appended
in a single step is at most $c\log(a\cdot b)+o(\log b)$, this would be enough
to bound the total length of a label by $c\log n+o(\log n)$ as desired. However,
it might happen that the fields appended in the same step consist of even
$\log n$ bits. We claim that in such case the size of the current tree decreases
more significantly, similarly to the analysis of the labeling scheme for routing
given in~\cite{thorup2001compact}. The following lemma captures this property.

\begin{lemma}
\label{lem:decrease}
Let $t$ denote the total length of all fields corresponding to a single step of the
recursion, and $s=4+c\log(a\cdot b)$. Then the size
of the current tree decreases by at least a factor of $b\cdot 2^{\max\{0,t-s\}}$.
\end{lemma}

\begin{proof}
If only the first field is defined, the claim is trivial, as its length is always at most $s$.
Observe that a node of degree $d$ must be mapped to a node of degree at least $d$
in the minor-universal tree, and consequently by Lemma~\ref{lem:weighted}
the length of the first field is at most $2+c\log(a\cdot b)-\log(1+\degree(u'))$,
where $u'$ is the nearest interesting node.
By construction, there are $\lceil \sum_{i=1}^{s}\weight(v_{i})/(n/b)\rceil$ virtual nodes
attached to $u'$, and so $\log(1+\degree(u')) \geq \log(\sum_{i=1}^{s}\weight(v_{i})/(n/b))$.
The length of the second field is $1+\log((\sum_{i=1}^{s}\weight(v_{i}))/\weight(v_{j}))$
(if $u'$ is a leaf or a branching node of $T'$, we define $s=1$ and $v_{1}=u'$).
Finally, the length of the third field is $1+\log k$ (or there is no third field).
All in all, we have the following:
\begin{align*}
t-s &\leq - \log(\sum_{i=1}^{s}\weight(v_{i})/(n/b)) + \log((\sum_{i=1}^{s}\weight(v_{i}))/\weight(v_{j}))+\log k  \\
&\leq \log(n/b) - \log(\weight(v_{j})) +\log k\\
&= \log((n/b)k/\weight(v_{j}))
\end{align*}
Observe that $\weight(v_{j}) \leq n/b$, so $b\cdot 2^{\max\{0,t-s\}} \leq b\cdot (n/b)k/\weight(v_j)$.
Finally, the size of the current tree changes to at most $\weight(v_{j})/k$ due to the
ordering of the children of $v_{j}$, or in other words decreases at least by a factor of $b \cdot (n/b)k/\weight(v_{j}) \geq b\cdot 2^{\max\{0,t-s\}}$.
\end{proof}

Now we analyze the total contribution of all steps to the total length of the label.
Let $t_{i}$ be the total length of all fields added in the $i$-th step, $r$ denote the
number of steps, and $s$ be defined as in Lemma~\ref{lem:decrease}. Using $c\geq 1$,
the total length of a label is then:
\[ \sum_{i=1}^{r} t_{i} \leq \sum_{i=1}^{r}s+\max\{0,t_{i}-s\} = r(4+ c\log a) + c\sum_{i=1}^{r}\log b+\max\{0,t_{i}-s\}  \]
Because in every step the size of the current tree decreases at least by a factor of
$b\cdot 2^{\max\{0,t_{i}-s\}}$, the product of such expressions is at most $n$, and
so the total length of a label can be upper bounded by:
\begin{align*}
\sum_{i=1}^{r} t_{i} &\leq \log n/\log b \cdot(4+c\log a) + c\log n  
= \Oh(\log n / \log b) +c\log n
\end{align*}
As long as $b=\omega(1)$, this is $c\log n+o(\log n)$ as required.

We move on to explaining how to implement a query in constant time given the labels
of $u$ and $v$. By considering the parts containing the concatenated fields, finding
the first position where they differ, and querying the associated rank/select structure,
we can determine in constant time the first field that is different in both labels. This
gives us the step of the recursive decomposition, such that $u$ and $v$ belong
to different small subtrees, or at least one of them is a big node (and thus does
not participate in further steps). Observe that the nearest common ancestor of $u$
and $v$ must a big node. Its label can be found by, essentially, truncating the
label of $u$ and $v$ and possibly appending a label obtained from the non-efficient
scheme. We now describe the details of this procedure.

Let $u'$ and $v'$ denote the nearest interesting ancestor of $u$ and $v$, respectively.
We would like to find the nearest common ancestor $w$ of $u'$ and $v'$. Note that,
by construction, $w$ must be an interesting node. Thus, we can use the 
minor-universal tree to obtain its label. However, the minor-universal tree does
not allow us to answer a query efficiently by itself. Thus, we preprocess all such queries in a table $\nca[x][y]$, where $x$ and $y$
are labels consisting of at most $2+c\log(a\cdot b)$ bits, and every entry also consists of at most $2+c\log(a\cdot b)$ bits
(to facilitate constant-time access, $\nca$ is stored as $(2+c\log(a\cdot b))^{2}$ separate tables
of the same size, one for each possible combination of $|x|$ and $|y|$, and every entry
is encoded with Elias $\gamma$ code~\cite{E75} and stored in a field of length $2(2+c\log(a\cdot b))$).
This lookup table is shared between all steps of the recursion. Now, if $w\notin \{u',v'\}$
then $w$ is the sought nearest common ancestor. Its label can be obtained by
truncating the label of, say, $u$, and appending a field storing the label of $w$
in the minor-universal tree. We also need to update the rank/select structure.
This can be also done by truncating and does not require adding a new integer to
the set, because we do not store the length of the last field explicitly.
Hence, the label of $w$ can be obtained in constant time with the standard
word-RAM operations. Otherwise, assume without losing the generality that $w=v'$.
If $w=u'$ also holds, then we look at the second field of both labels (if there
is none in one of them then again $w$ is the sought nearest common ancestor).
If there are equal then the nearest big ancestor of $u$ and $v$ is the same,
and should be returned as the nearest common ancestor. Its label can
be obtained by truncating the label of either $u$ or $v$.
Otherwise, recall that $w$ has at most one big child, and so there is a path
$v_{0}-v_{1}-\ldots v_{s}=w$ between two interesting nodes, such that
$u$ belongs to a small subtree attached to some $v_{i}$ and $v$ is in a small
subtree attached to some $v_{j}$, where $i,j\in \{1,2,\ldots,s\}$.
Because the binary codes assigned to the nodes of the path preserve the bottom-top
order, we can check whether $i<j$, $i=j$, or $i>j$. If $i>j$ ($i<j$), then $v_{i}$ ($v_{j}$)
is the sought nearest common ancestor, and its label can be obtained by truncating
the label of $u$ ($v$). Finally, if $i=j$, then the nearest common
ancestor must be $v_{i}=v_{j}$, because we know that $u$ and $v$ do not belong
to the same small subtree, and truncate the label of either $u$ or $v$.

We analyze the total length of a label. It consists of 1) the concatenated fields,
2) rank/select structure encoding the lengths of the fields, 3) a lookup table
for answering queries in the minor-universal tree. The total length of all the
fields is $L=\Oh(\log n / \log b)+c\log n$. The rank/select structure
from Lemma~\ref{lem:rank/select} is built for a set of at most $\log n / \log b$
integers from $[L]$,
and so takes $\Oh(\log n / \log b \cdot \log (L \cdot \log b / \log n))$ bits of space.
The lookup table uses $(a\cdot b)^{2c}\cdot 2(2+c\log (a\cdot b))^{3}$ bits of space, making the total length:
\[ c\log n+\Oh(\log n \cdot \log \log b / \log b + (a\cdot b)^{2c}\log^{3} b)\]
By setting $b=1/a \cdot (\log n)^{1/(4c)}$ we obtain labels of length $c\log n + o(\log n)$
and constant decoding time.

\begin{theorem}
\label{thm:transform}
Consider any class $\mathcal{T}$ of rooted trees closed under taking topological minors.
If, for any $n$, there exists a minor-universal tree of size $n^{c}$ then there exists
a labeling scheme for nearest common ancestors with labels consisting
of $c\log n+o(\log n)$ bits and constant query time.
\end{theorem}


\bibliographystyle{plain} 
\bibliography{biblio}

\begin{thebibliography}{10}

\bibitem{AbboudGMW17}
Amir Abboud, Pawel Gawrychowski, Shay Mozes, and Oren Weimann.
\newblock Near-optimal compression for the planar graph metric.
\newblock {\em CoRR}, abs/1703.04814, 2017.

\bibitem{abiteboul2006compact}
Serge Abiteboul, Stephen Alstrup, Haim Kaplan, Tova Milo, and Theis Rauhe.
\newblock Compact labeling scheme for ancestor queries.
\newblock {\em SIAM Journal on Computing}, 35(6):1295--1309, 2006.

\bibitem{AlonN17}
Noga Alon and Rajko Nenadov.
\newblock Optimal induced universal graphs for bounded-degree graphs.
\newblock In {\em 28th SODA}, pages 1149--1157, 2017.

\bibitem{alstrup2005labeling}
Stephen Alstrup, Philip Bille, and Theis Rauhe.
\newblock Labeling schemes for small distances in trees.
\newblock {\em SIAM Journal on Discrete Mathematics}, 19(2):448--462, 2005.

\bibitem{alstrup2015optimal}
Stephen Alstrup, S{\o}ren Dahlgaard, and Mathias B{\ae}k~Tejs Knudsen.
\newblock Optimal induced universal graphs and adjacency labeling for trees.
\newblock In {\em 56th FOCS}, pages 1311--1326, 2015.

\bibitem{alstrup2016simpler}
Stephen Alstrup, Cyril Gavoille, Esben~Bistrup Halvorsen, and Holger Petersen.
\newblock Simpler, faster and shorter labels for distances in graphs.
\newblock In {\em 27th SODA}, pages 338--350, 2016.

\bibitem{AlstrupGKR02}
Stephen Alstrup, Cyril Gavoille, Haim Kaplan, and Theis Rauhe.
\newblock Nearest common ancestors: a survey and a new distributed algorithm.
\newblock In {\em 14th SPAA}, pages 258--264, 2002.

\bibitem{alstrup2015distance}
Stephen Alstrup, Inge~Li G{\o}rtz, Esben~Bistrup Halvorsen, and Ely Porat.
\newblock Distance labeling schemes for trees.
\newblock In {\em 43rd ICALP}, pages 132:1--132:16, 2016.

\bibitem{AHL14}
Stephen Alstrup, Esben~Bistrup Halvorsen, and Kasper~Green Larsen.
\newblock Near-optimal labeling schemes for nearest common ancestors.
\newblock In {\em 25th SODA}, pages 972--982, 2014.

\bibitem{alstrup2015adjacency}
Stephen Alstrup, Haim Kaplan, Mikkel Thorup, and Uri Zwick.
\newblock Adjacency labeling schemes and induced-universal graphs.
\newblock In {\em 47th STOC}, pages 625--634, 2015.

\bibitem{alstrup2002small}
Stephen Alstrup and Theis Rauhe.
\newblock Small induced-universal graphs and compact implicit graph
  representations.
\newblock In {\em 43rd FOCS}, pages 53--62, 2002.

\bibitem{BenderF00}
Michael~A. Bender and Martin Farach{-}Colton.
\newblock The {LCA} problem revisited.
\newblock In {\em 4th LATIN}, pages 88--94, 2000.

\bibitem{bonichon2007short}
Nicolas Bonichon, Cyril Gavoille, and Arnaud Labourel.
\newblock Short labels by traversal and jumping.
\newblock {\em Electronic Notes in Discrete Mathematics}, 28:153--160, 2007.

\bibitem{Clark}
David~Richard Clark.
\newblock {\em Compact Pat Trees}.
\newblock PhD thesis, University of Waterloo, 1998.

\bibitem{E75}
Peter Elias.
\newblock Universal codeword sets and representations of the integers.
\newblock {\em {IEEE} Transactions on Information Theory}, 21(2):194--203,
  1975.

\bibitem{fischer2009short}
Johannes Fischer.
\newblock Short labels for lowest common ancestors in trees.
\newblock In {\em 17th ESA}, pages 752--763, 2009.

\bibitem{fraigniaud2010compact}
Pierre Fraigniaud and Amos Korman.
\newblock Compact ancestry labeling schemes for xml trees.
\newblock In {\em 21st SODA}, pages 458--466, 2010.

\bibitem{FGNW16}
Ofer Freedman, Pawel Gawrychowski, Patrick~K. Nicholson, and Oren Weimann.
\newblock Optimal distance labeling schemes for trees.
\newblock {\em CoRR}, abs/1608.00212, 2016.

\bibitem{gavoille2007distributed}
Cyril Gavoille and Arnaud Labourel.
\newblock Distributed relationship schemes for trees.
\newblock In {\em 18th ISAAC}, pages 728--738, 2007.
\newblock Announced at PODC'07.

\bibitem{gavoille2004distance}
Cyril Gavoille, David Peleg, St{\'e}phane P{\'e}rennes, and Ran Raz.
\newblock Distance labeling in graphs.
\newblock {\em Journal of Algorithms}, 53(1):85--112, 2004.
\newblock A preliminary version in \emph{12th SODA}, 2001.

\bibitem{GawrychowskiKU16}
Pawel Gawrychowski, Adrian Kosowski, and Przemyslaw Uznanski.
\newblock Sublinear-space distance labeling using hubs.
\newblock In {\em 30th DISC}, pages 230--242, 2016.

\bibitem{GawrychowskiU16}
Pawel Gawrychowski and Przemyslaw Uznanski.
\newblock A note on distance labeling in planar graphs.
\newblock {\em CoRR}, abs/1611.06529, 2016.

\bibitem{HarelT84}
Dov Harel and Robert~Endre Tarjan.
\newblock Fast algorithms for finding nearest common ancestors.
\newblock {\em {SIAM} J. Comput.}, 13(2):338--355, 1984.

\bibitem{HrubesWY10}
Pavel Hrubes, Avi Wigderson, and Amir Yehudayoff.
\newblock Relationless completeness and separations.
\newblock In {\em 25th CCC}, pages 280--290, 2010.

\bibitem{Kannan}
Sampath Kannan, Moni Naor, and Steven Rudich.
\newblock Implicit representation of graphs.
\newblock {\em SIAM Journal on Discrete Mathematics}, 5(4):596--603, 1992.

\bibitem{KatzKKP04}
Michal Katz, Nir~A. Katz, Amos Korman, and David Peleg.
\newblock Labeling schemes for flow and connectivity.
\newblock {\em {SIAM} J. Comput.}, 34(1):23--40, 2004.

\bibitem{Peleg00}
David Peleg.
\newblock Proximity-preserving labeling schemes.
\newblock {\em Journal of Graph Theory}, 33(3):167--176, 2000.

\bibitem{Peleg05}
David Peleg.
\newblock Informative labeling schemes for graphs.
\newblock {\em Theor. Comput. Sci.}, 340(3):577--593, 2005.

\bibitem{petersen2015near}
Casper Petersen, Noy Rotbart, Jakob~Grue Simonsen, and Christian Wulff-Nilsen.
\newblock Near-optimal adjacency labeling scheme for power-law graphs.
\newblock In {\em 43rd ICALP}, pages 133:1--133:15, 2016.

\bibitem{rotbart2016new}
Noy~Galil Rotbart.
\newblock {\em New Ideas on Labeling Schemes}.
\newblock PhD thesis, University of Copenhagen, 2016.

\bibitem{thorup2001compact}
Mikkel Thorup and Uri Zwick.
\newblock Compact routing schemes.
\newblock In {\em 13th SPAA}, pages 1--10, 2001.

\bibitem{YoungCW99}
Fung~Yu Young, Chris C.~N. Chu, and D.~F. Wong.
\newblock Generation of universal series-parallel boolean functions.
\newblock {\em J. {ACM}}, 46(3):416--435, 1999.

\end{thebibliography}

\newpage
\appendix

\section{Ordered Trees}
\label{sec:ordered}

Hrubes et al.~\cite{HrubesWY10} consider ordered binary trees and construct a minor-universal
tree of size $\Oh(n^{4})$ for ordered binary trees on $n$ nodes.
We modify their construction to obtain a smaller minor-universal tree for ordered binary
trees $B'_{n}$ as described below.

We invoke Lemma~\ref{lem:order} with $N=\lfloor(1-\alpha)n\rfloor$ to obtain a sequence
$a_{\lfloor(1-\alpha)n\rfloor}=(a(1),a(2),\ldots,a(k))$. Then, $B'_{n}$ consists of a path
$u_{1}-v_{1}-u_{2}-v_{2}-\ldots -u_{k+1}-v_{k+1}-w$. For every $i=1,2,\ldots,k$, we attach
a copy of $B'_{a(i)-1}$ as the left child of $u_{i}$, and we also attach a copy of 
$B'_{a(i)-1}$ as the right child of $v_{i}$.
Additionally, we attach a copy of $B'_{\lfloor\alpha\cdot n\rfloor}$ as the left child of $w$, and 
another copy of $B'_{\lfloor\alpha\cdot n\rfloor}$ as the right child of $w$.
By a similar argument to the one used to argue that $B_{n}$ is a minor-universal tree for all binary
trees on $n$ nodes we can show that $B'_{n}$ is a minor-universal tree for all ordered binary
trees on $n$ nodes if $\alpha\geq 0.5$. Its size can be bounded as follows:
\begin{align*}
|B'_{n}| &= 1 + 2|B'_{\lfloor\alpha\cdot n\rfloor }| + 2\sum_{i=0}^{\lfloor\log(1-\alpha)n\rfloor} 2^{i} \cdot (1+|B'_{\lfloor (1-\alpha)n/2^{i}\rfloor-1}|) 
\end{align*}
To show that $|B'_{n}| \leq n^{c}$ it is enough that the following inequality holds:
\begin{align*}
2\alpha^{c} + (1-\alpha)^{c} \cdot 2^{c} / (2^{c-1}-1) &\leq 1
\end{align*}
So it is enough that $2(A/(1+A))^{c}+(1/A)^{c}\cdot 2^{c}/(2^{c-1}-1) \leq 1$, where $A=(2^{c-1}/(2^{c-1}-1))^{1/(c-1)}$.
This can be verified to hold for $c=2.331$ by choosing $A=1.463$ and $\alpha=0.594$.

\section{Universal Trees of Young et al.}
\label{sec:universal}

To present the definition of a universal tree in the sense of Young et al.~\cite{YoungCW99}
we first need to present their original definition of two operations on trees:
\begin{description}
\item[Cutting.] Two nodes $a$ and $b$, such that $a$ is a child of $b$, are selected. The entire subtree
rooted at $a$ and the edge between $a$ to $b$ are removed.
\item[Contraction.] An internal node $b$, which has parent $a$ and a single child $c$, is selected. Node
$b$ is removed. If $c$ is internal node, the children of $c$ are made children of $a$ and $c$ is removed.
If $c$ is a leaf, it becomes a child of $a$.
\end{description}
Then, tree $T$ implements tree $T'$ if $T'$ can be obtained by applying a sequence of
cutting and contraction operations to $T$. Finally, $T$ is an $n$-universal tree if it can implement
any tree $T'$ with at most $n$ leaves and no degree-1 nodes. Notice that the degrees of the nodes
of $T'$ are not bounded in the original definition. However, for our purposes it will be
enough to consider binary trees. We want to prove a lower bound on the number of leaves
of an $n$-universal tree.

We introduce the notion of parity-preserving
minor-universal trees. We say that $T$ is a parity-preserving minor-universal tree for a class $\mathcal{T}$
of rooted trees if, for any $T'\in\mathcal{T}$, the nodes of $T'$ can be mapped to the nodes of $T$
as to preserve the NCA relationship and the parity of the depth of every node.

\begin{lemma}
\label{lem:preserving}
An $n$-universal tree is a parity-preserving minor-universal tree for binary trees on $n$ leaves and no degree-1 nodes.
\end{lemma}

\begin{proof}
We first observe that the definition of contraction can be changed as follows:
\begin{description}
\item[Contraction.] An internal node $b$, which has parent $a$ and a single child $c$, is selected. Node
$b$ is removed. The children of $c$ are made children of $a$ and $c$ is removed.
\end{description}
This is because if $c$ is a leaf, reattaching $c$ to $a$ and removing $b$ is equivalent to cutting $c$.

Let $T$ be an $n$-universal tree and consider any binary tree $T'$ on $n$ leaves and no degree-1 nodes.
By assumption, $T$ implements $T'$, so we can obtain $T'$ from $T$ by a sequence of cutting and contraction
operations. We claim that if $T$ implements $T'$ then the nodes of $T'$ can be mapped to the nodes of $T$
as to preserve the NCA relationship and the parity of the depth of every node. We prove this by induction
on the length of the sequence. If the sequence is empty, the claim is obvious. Otherwise, assume that
$T_{1}$ is obtained from $T$ by a single cutting or contraction, and by the inductive assumption
the nodes of $T'$ can be mapped to the nodes of $T_{1}$ as to preserve the NCA relationship and
the parity of the depth of every node. For cutting, the claim is also obvious, as we can use the same
mapping. For contraction, we might need to modify it. Observe that at most one node $u\in T'$ is mapped
to the node $a$. If there is no such node, or the degree of $u$ in $T'$ is 1, we are done because 
the parity of the depth of every node that appears in both $T$ and $T_{1}$ is the same and the NCA
relationship in $T$ restricted to the nodes that appear in $T_{1}$ is also identical.
Otherwise, let $v_{1}$ and $v_{2}$ be the children of $u$ in $T'$.
$v_{1}$ ($v_{2}$) is mapped to a node in the subtree rooted at a child $a_{1}$ ($a_{2}$) of $a$ in $T_{1}$. If both
$a_{1}$ and $a_{2}$ are children of $c$ in $T$ then we modify the mapping so that
$u$ is mapped to $c$ in $T$, and otherwise $u$ is mapped to the original $a$ in $T$. Mapping of other nodes of
$T'$ remains unchanged. It can be verified that the obtained mapping indeed preserves the NCA relationship
and the parity of the depth of every node. Thus, $T$ is indeed a parity-preserving minor-universal tree
for binary trees on $n$ leaves and no degree-1 nodes.
\end{proof}

We need one more definition. Let $\inner(T')$ be the set of inner nodes of a tree $T'$.
A tree $T$ is a parity-constrained minor-universal tree for a class $\mathcal{T}$
of rooted trees if, for any $T'\in \mathcal{T}$ and any assignment $c : \inner(T') \rightarrow \{0,1\}$,
the nodes of $T'$ can be mapped to the nodes of $T$ as to preserve the NCA relationship and,
for any $v\in \inner(T')$, if $c(v)=0$ then $v$ is mapped to a node at even depth and if $c(v)=1$ then $v$
is mapped to a node at odd depth. $c$ is called the parity constraint.

\begin{lemma}
\label{lem:constrained}
A $(2n-1)$-universal tree is a parity-constrained minor-universal
tree for binary trees on $n$ leaves and no degree-1 nodes.
\end{lemma}

\begin{proof}
Given a tree $T'$ on $n$ leaves and no degree-1 nodes, and an assignment $c: \inner(T') \rightarrow \{0,1\}$,
we will construct a tree $T''$ on at most $2n-1$ leaves (and also no degree-1 nodes), such that if the nodes
of $T''$ can be mapped to the nodes of $T$ as to preserve the NCA relationship and the parity of the depth
of every node, then the nodes of $T'$ can be mapped to the nodes of $T$ as to preserve the NCA relationship
and respect the parity constraint. Together with Lemma~\ref{lem:preserving}, this proves the lemma.

We transform $T'$ into $T''$ as follows. We consider all inner nodes of $T'$ in the depth-first order.
Let $r$ be the root of $T'$. If $c(r)=1$, then we create a new root $r'$, make $r$ a child of $r'$,
and attach a new leaf as another child of $r'$.
For a node $v$ with parent $u$ in $T'$, if $c(u)\neq c(v)$ then we do nothing. If $c(u)= c(v)$,
then we attach a new child $v'$ to $u$, make $v$ a child of $v'$, and attach a new leaf as another child
of $v'$. The total number of new leaves created during the process is at most the number of inner
nodes of $T'$, so the total number of leaves in $T''$ is at most $2n-1$. It is easy to see that preserving the parity of the
depth of every node of $T''$ implies respecting the parity constraint for the original nodes of $T'$,
and the NCA relationship restricted to the original nodes in $T''$ is the same as in $T'$.
\end{proof}

We are ready to proceed with the main part of the proof. Our goal is to lower bound the number of leaves
$b(n)$ in a parity-constrained minor-universal tree for binary trees on $n$ leaves and no degree-1 nodes.
By Lemma~\ref{lem:constrained}, this also implies a lower bound on the number of leaves (and thus the size)
of a universal tree in the sense of Young et al.~\cite{YoungCW99}. Because a parity-constrained minor-universal
tree is a minor-universal tree, we could simply apply Lemma~\ref{lem:lowerbinary} and conclude that
$b(n)=\Omega(n^{1.728})$. Our goal is to obtain a stronger lower bound by exploiting the parity constraint.

\begin{lemma}
\label{lem:loweryoung}
$b(n) \geq 1+2\sum_{s\geq 2} b(\lfloor n/s \rfloor).$
\end{lemma}

\begin{proof}
Let $T$ be a parity-constrained minor-universal tree for binary trees on $n$ leaves and no degree-1 nodes.
We choose $d\in\{0,1\}$ such that at most half of nodes of degree 2 or more in $T$ is at depth congruent
to $d$ modulo 2. 

For any $s\geq 2$, we define an $s$-caterpillar and $s(v)$ for any node $v\in T$
as in the proof of Lemma~\ref{lem:lowerbinary}, except that now we require that
all inner nodes of the $s$-caterpillar should be mapped to nodes at depth congruent
to $d$ modulo 2. This changes the properties of $s(v)$ as follows:
\begin{enumerate}
\item For every child $u$ of $v$, $s(u) \leq s(v)$.
\item If the degree of $v$ is 1 then, for the unique child $u$ of $v$, $s(u)=s(v)$.
\item If the degree of $v$ is at least 2 and the depth of $v$ is not congruent to $d$ modulo 2 then, for some child $u$ of $v$, $s(u)=s(v)$.
\item If the degree of $v$ is at least 2 and the depth of $v$ is congruent to $d$ modulo 2 then, for some child $u$ of $v$, $s(u)=s(v)-1$.
\end{enumerate}
Then, for any $s\geq 2$,
we consider any binary tree on $\lfloor n/s\rfloor$ leaves and no degree-1 nodes,
and any choice of the parities for all of its inner nodes.
We replace all leaves of the original binary tree by $s$-caterpillars and require that
their inner nodes are mapped to nodes at depth congruent to $d$ modulo 2,
while for the original inner nodes the required parity remains unchanged.
By assumption, it must be possible to map the nodes of the new binary tree
to the nodes of $T$ as to preserve the NCA relationship and respect the parity
constraint. The leaves of the original binary tree must be mapped to nodes
on level at least $s$ in $T$.
As in the proof of Lemma~\ref{lem:lowerbinary},
we obtain a tree $T'$ by removing all nodes on level smaller than $s$
from $T$. From the properties of $s(v)$ it is clear that every leaf of $T'$
is on level exactly $s$ and of degree at least 2. We claim that, additionally, all leaves of $T'$ are
at depth congruent to $d$ modulo 2. This is because if a node $v\in T$
is at depth not congruent to $d$ modulo 2 then, for some child $u$ of $v$,
$s(u)=s(v)$, so in fact $v$ cannot be a leaf in $T'$. For any binary tree on $\lfloor n/s\rfloor$
leaves and no degree-1 nodes and any choice of the parities for the inner nodes,
the nodes of the binary tree can be mapped to the nodes of $T'$ as to preserve
the NCA relationship and respect the parity constraint. By lower bounding the
number of leaves in $T'$ we thus obtain that the number of degree-2 nodes on
level $s$ and at depth congruent to $d$ modulo 2 in $T$ is at least $b(\lfloor n/s\rfloor)$.
Thus, the total number of degree-2 nodes at depth congruent to $d$ modulo 2
in $T$ is
\[ \sum_{s\geq 2} b(\lfloor n/s\rfloor) \]
Finally, by the choice of $d$ the total number of degree-2 nodes is at least twice
as large, and so the total number of leaves exceeds
\[ 2\sum_{s\geq 2} b(\lfloor n/s\rfloor) \qedhere\]
\end{proof}

To extract an explicit lower bound from Lemma~\ref{lem:loweryoung}, we proceed
as in Theorem~\ref{thm:binary}. It is straightforward to verify that the same reasoning
can be used to show that, if $\zeta(c)>1.5$ then $b(n)=\Omega(n^{c})$. We verify
that $\zeta(2.185)>1.5$, and so $b(n)=\Omega(n^{2.185})$.

\section{Missing Proofs}
\label{sec:proofs}

\rankselectlemma*

\begin{proof}
The encoding is similar to Lemma 2.2 of~\cite{FGNW16}, except that we cannot use a black box predecessor structure.
Let $L=s\cdot \max\{1,\log \frac{M}{s} \}$ and the set consists of $x_{1}<x_{2}<\ldots<x_{s'}$, where $s'\leq s$.

We partition the universe $[1,M]$ into blocks of length $b=\frac{M}{s}$. The encoding starts with
$b$ encoded with the Elias $\gamma$ code~\cite{E75}. Then we store every $x_{i}\bmod b$ using
$2+\log b$ bits. The encodings of $x_{i}\bmod b$ are separated by single 1s.
This takes $\Oh(\log b+s+s\log b)=L$ bits so far. We need to also store every $y_{i}=x_{i}\bdiv b$.
We observe that $0\leq y_{1}\leq y_{2}\leq \ldots y_{s'}\leq s$. Hence, we can encode them with a bit vector
of length at most $2s$, which is a concatenation of $0^{y_{i}-y_{i-1}}1$ for $i=1,2,\ldots,s'$. The bit vector
is augmented with a select structure of Clark~\cite[Chapter~2.2]{Clark}, which uses $o(s)$ additional
bits and allows us to extract the $i^\text{th}$ bit set to 1 in constant time. This all takes $\Oh(L)$ bits
of space and allows us to decode any $x_{i}$ in constant time by extracting $x_{i}\bmod b$ and $x_{i}\bdiv b$.

To find the successor of $x$, we first compute $y=x\bdiv b$. Then, using the bit vector we can find
in constant time the maximal range of integers $x_{i},x_{i+1},\ldots,x_{j}$ such $x_{k} \bdiv b=y$ for every
$k=i,i+1,\ldots,j$. The successor can be then found by finding the successor of $x\bmod b$ among
$x_{i}\bmod b,x_{i+1}\bmod b,\ldots,x_{j}\bmod b$ and, if there is none, returning $x_{j+1}$.
To find the successor of $x\bmod b$ in the range, we use the standard method of repeating the encoding
of $x\bmod b$ separating by single 0s $(j-i+1)$ times by multiplying with an appropriate constant
(that can be computed with simple arithmetical operations in constant time, assuming that we can multiply and compute
a power of 2 in constant time), and then subtracting the obtained bit vector from
a bit vector containing the encodings of $x_{i}\bmod b,x_{i+1}\bmod b,\ldots,x_{j}\bmod b$ separated
by single 1s (that is obtained from the stored encoding with standard bitwise operations). 
The bit vectors fit in a constant number of words, and hence all operations can be implemented in constant time.

Finally, we describe how to truncate the encoding. The only problematic part is that we have used
a black box select structure. Now, we want to truncate the stored bit vector, and this might change
the additional $o(s)$ bits. We need to inspect the internals of the structure.

Recall that the structure of Clark~\cite[Chapter~2.2]{Clark} for selecting the $k^\text{th}$ occurrence of 1
partitions a bit vector of length $m$ into macroblocks by choosing every $t_{1}^\text{th}$ such occurrence, where
$t_{1}=\log m \log \log m$. We encode every macroblock separately and concatenate their encodings.
Additionally, for every $i$ we store the starting position of the $i$-th macroblock in the bit vector and the
starting position of its part of the encoding in an array using $\Oh(m/t_{1}\cdot \log m)=\Oh(m/\log\log m)$ bits.
Now consider a single macroblock and let $r$ be its length. If $r > t_{1}^{2}$, we store the position
of every 1 inside the macroblock explicitly. This is fine because there can be at most $m/t_{1}^{2}$ such blocks, so this
takes $\Oh(m/t_{1}^{2}\cdot t_{1}\cdot \log m) = \Oh(m/\log\log m)$ bits. Otherwise, we will
encode the relative position of every 1, but not explicitly. We further partition such macroblock into blocks
by choosing every $t_{2}^\text{th}$ occurrence of 1, where $t_{2}=(\log\log m)^{2}$. We encode
every block separately and concatenate their encodings, and for every $i$ store the relative starting
position of the $i$-th block (in its macroblock) and the relative starting position of its part of the encoding
(in the encoding of the macroblock) in an array using $\Oh(m/t_{2}\cdot \log \log m)$ bits (we will make sure
that the encoding of any macroblock takes only $\Oh(\polylog m)$ bits). Then, let again $r$ be length
of a block. If $r > t_{2}^{2}$, we can store the relative position of every 1 (in its macroblock) inside the block
explicitly. Otherwise, the whole block is of length less than $t_{2}^{2}< \frac{1}{2}\log m$, and we can
tabulate. In more detail, for every bit vector of length at most $\frac{1}{2}\log m$ (there are $\Oh(\sqrt{m})$
of them), we store the positions of the at most $\frac{1}{2}\log m$ 1s explicitly. This precomputed table
takes $o(m)$ bits, so can be stored as a part of the structure.
Then, given a block of length less than $\frac{1}{2}\log m$, we extract its corresponding fragment
of the bit vector using the standard bitwise operations, and use the precomputed table.

We are now ready to describe how to update the select structure after truncating the bit vector
after the $k^\text{th}$ occurrence of 1. We first determine the macroblock containing this occurrence,
say that it is the $i^\text{th}$ macroblock. We can easily discard all further macroblocks by checking where
the encoding of the $(i+1)^\text{th}$ macroblock starts and erasing everything starting from there.
We also erase the starting positions stored for all further macroblocks, and move the encoding
just after the remaining starting positions. This can be done in constant time using standard bitwise
operations. Then, we inspect the $i^\text{th}$ macroblock. If the positions of all 1s are stored
explicitly, we erase a suffix of this sequence. This is now problematic, because maybe after erasing
a suffix $r$ becomes at most $t_{2}^{2}$ and we actually need the other encoding. We overcome this difficulty
by changing the definition: a macroblock is partitioned into a prefix of length $t_{2}^{2}$ and the remaining
suffix. The occurrences of all 1s in the suffix are stored explicitly, and we also store the number of occurrences
in the prefix. Then, the prefix is partitioned into blocks by choosing every $t_{2}^\text{th}$ occurrence. 
To truncate the prefix, we need to completely erase a suffix of blocks, which can be done in constant
time, and modify the last remaining block. If the encoding of the last block consists of explicitly
stored relative positions, we just need to erase its suffix, which again can be done in constant time.
Otherwise, there is actually nothing to do. Additionally, we need to make sure that the precomputed
table does not have to be modified. To this end, instead of tabulating every bit vector of length at
most $\frac{1}{2}\log m$, we tabulate every bit vector of length at most $\frac{1}{2}\log s$
(instead of $\frac{1}{2}\log m$).
\end{proof}

\end{document}